\documentclass[11pt, reqno, a4paper]{amsart}
\usepackage{amsfonts}

\usepackage[english]{babel}
\usepackage[latin1]{inputenc}
\usepackage{graphics}
\usepackage{ulem}
\usepackage{hhline}
\usepackage{dsfont}
\usepackage{mathrsfs}
\usepackage{fancyhdr}
\usepackage{amsmath, amssymb}
\usepackage{rotating}
\usepackage[latin1]{inputenc}
\usepackage[T1]{fontenc}
\usepackage{fancybox}
\usepackage{color}
\usepackage{colortbl}
\usepackage{setspace}
\usepackage{enumerate}
\usepackage[nice]{nicefrac}
\usepackage{amsthm}
\usepackage{multicol}
\usepackage{pifont}
\usepackage[french]{minitoc}
\usepackage{latexsym, amsfonts}
\usepackage{amsthm}
\usepackage{varioref}
\usepackage{textcomp}
\usepackage{lmodern}
\usepackage{mathpazo}
\usepackage{euscript}
\usepackage[pdftex]{hyperref}

\numberwithin{equation}{section}
\makeatletter\@addtoreset{equation}{section}

\DeclareMathSymbol{\subsetneqq}{\mathbin}{AMSb}{36}
\newtheorem {theorem}{Theorem}[section]

\newtheorem {definition}[theorem]{Definition}

\newtheorem {proposition}[theorem]{Proposition}
\newtheorem {remark}[theorem]{Remark}

\newtheorem {corollary}[theorem]{Corollary}

\newcommand{\C}{\mathbb C}

\begin{document}
\author{A. Ghanmi, A. Hafoud, Z. Mouayn}
\title[Generalized Binomial Probability]{Generalized Binomial Probability Distributions attached to Landau levels on
the Riemann sphere}
\address{{(A.G. $\&$ A.H.)} \\
Department of Mathematics, Faculty of Sciences, P.O. Box 1014\\
Mohammed V University, Agdal, 10 000 Rabat - Morocco}
\email{(A.G.) {allalghanmi@gmail.com}, (A.H.) {hafoud$_{-}$ali@yahoo.fr}}
\address{{(Z.M.)} \\
Department of Mathematics, Faculty of Technical Sciences, Sultan Moulay
Slimane University, P.O. Box 523, 23000 Béni Mellal, Morocco.}
\email{mouayn@gmail.com}
\date{\today}
\maketitle

\begin{abstract}
A family of generalized binomial probability distributions attached to
Landau levels on the Riemann sphere is introduced by constructing a kind of
generalized coherent states. Their main statistical parameters are obtained
explicitly. As application, photon number statistics related to coherent
states under consideration are discussed.
\end{abstract}

\section{Introduction}

The\textit{\ binomial states} (BS) are the field states that are
superposition of the number states with appropriately chosen coefficients
\cite{Stoler84}. Precisely, these states are labeled by points $z$ of the
Riemann sphere $\mathbb{S}^{2}\equiv \mathbb{C}\cup \left\{ \infty \right\}$%
, and are of the form
\begin{equation}
\mid z\text{, }B>=\left( 1+\left| z\right| ^{2}\right)
^{-B}\sum\limits_{j=0}^{2B}\left( \frac{\left( 2B\right) !}{j!\left(
2B-j\right) !}\right) ^{\frac{1}{2}}z^{j}\mid j>  \label{1.1}
\end{equation}
where $B\in \mathbb{Z}_{+}$ is a fixed integer parameter and $\mid j>$ are
number states of the field mode.

Define $\mu _{_{z}}$ to be $\mu _{_{z}}:=\left| z\right| ^{2}\left( 1+\left|
z\right| ^{2}\right) ^{-1}$. Then the probability for the production of $j$
photons is given by the squared modulus of the projection of the BS  $\mid
z,B>$  onto the number state $\mid j>$ as
\begin{equation}
\left| \left\langle j\mid z,B\right\rangle \right| ^{2}=\frac{(2B)!}{%
j!(2B-j)!}\mu _{_{z}}^{j}\left( 1-\mu _{_{z}}\right) ^{2B-j}.  \label{1.2}
\end{equation}
The latter is recognized as the \ binomial probability density $\mathcal{B}%
\left( 2B,\mu _{_{z}}\right) $ where $\left\{ \mu _{_{z}},\left( 1-\mu
_{_{z}}\right) \right\} $ are the probabilities of the two possible outcomes
of a Bernoulli trial \cite{Feller}.

Also, observe that the coefficients in the finite superposition of number
states in \eqref{1.1}:
\begin{equation}
h_{j}^{B}(z) :=\left( 1+\left| z\right| ^{2}\right) ^{-B}\left( \frac{\left(
2B\right) !}{j!\left( 2B-j\right) !}\right) ^{\frac{1}{2}}z^{j},\qquad
j=0,1,2,\cdots ,2B,  \label{1.3}
\end{equation}
constitutes an orthonormal basis of the null space
\begin{equation}
\mathcal{A}_{B}(\mathbb{S}^{2}):=\left\{ \varphi \in L^{2}(\mathbb{S}%
^{2}),\quad H_{B}\left[ \varphi \right] =0\text{ }\right\}  \label{1.4}
\end{equation}
of the second-order differential operator
\begin{equation}
H_{B}:=-\left( 1+\left| z\right| ^{2}\right) ^{2}\frac{\partial ^{2}}{%
\partial z\partial \overline{z}}-B\left( 1+\left| z\right| ^{2}\right)
\left( z\frac{\partial }{\partial z}-\overline{z}\frac{\partial }{\partial
\overline{z}}\right) +B^{2}\left| z\right| ^{2}-B,  \label{1.5}
\end{equation}
which constitutes (in suitable units and up to additive constant) a
realization in $L^{2}(\mathbb{S}^{2})$ of the Schrödinger operator with
uniform magnetic field on $\mathbb{S}^{2}$, with a field strength
proportional to $B$ (see \cite{FeraVese}). The given orthonormal basis $%
h_{j}^{B}(z)$ together with the reproducing kernel
\begin{equation}
K_{B}\left( z,w\right) =(2B+1)\left( 1+z\overline{w}\right) ^{2B}\left(
1+\left| z\right| ^{2}\right) ^{-B}\left( 1+\left| w\right| ^{2}\right) ^{-B}
\label{1.6}
\end{equation}
of the Hilbert space $\mathcal{A}_{B}(\mathbb{S}^{2})$ can be used to
interpret the projection of the BS $\mid z,B>$ onto the number state $\mid
j> $ mentioned in \eqref{1.2} by writing
\begin{equation}
\left\langle j\mid z,B\right\rangle =\left( K_{B}\left( z,z\right) \right)
^{-\frac{1}{2}}h_{j}^{B}(z).  \label{1.7}
\end{equation}
The space $\mathcal{A}_{B}(\mathbb{S}^{2})$ is nothing else than the
eigenspace associated with the first eigenvalue of the spectrum of $\ H_{B}$
acting on $L^{2}(\mathbb{S}^{2})$, which consists of an infinite set of
eigenvalues \textit{(spherical Landau levels) } of the form:
\begin{equation}
\epsilon _{m}^{B}:=\left( 2m+1\right) B+m\left( m+1\right) ,\qquad
m=0,1,2,\cdots ,  \label{1.8}
\end{equation}
with finite multiplicity; i.e., the associated $L^{2}-$eigenspace
\begin{equation}
\mathcal{A}_{B,m}(\mathbb{S}^{2}):=\left\{ \varphi \in L^{2}(\mathbb{S}%
^{2}),\quad H_{B}\left[ \varphi \right] =\epsilon _{m}^{B}\varphi \right\}
\label{1.9}
\end{equation}
is of finite dimension equals to $d_{B,m}:=2B+2m+1.$

Here, we take the advantage of the fact that each of the eigenspaces in %
\eqref{1.9} admits an orthogonal basis denoted $h_{j}^{B,m}(z)$, $%
j=0,1,2,\cdots ,2B+2m$, whose elements are expressed in terms of Jacobi
polynomials $P_{\eta }^{\left( \tau ,\varsigma \right) }(.)$, as well as a
reproducing kernel $K_{B,m}\left( z,w\right) $ in an explicit form (see \cite
{PZ}) to consider a set of coherent states by adopting a generalized coherent
states technique \textit{'' à la Iwata''} \cite{Iwata} as:
\begin{equation}
\mid z,B,m>=\left( K_{B,m}\left( z,z\right) \right) ^{-\frac{1}{2}%
}\sum\limits_{j=0}^{2B+2m}\frac{h_{j}^{B,m}(z)}{\sqrt{\rho _{B,m}\left(
j\right) }}\mid j>,  \label{1.10}
\end{equation}
where $\rho _{B,m}\left( j\right) $ denotes the norm square of $\
h_{j}^{B,m}(z)$ in $L^{2}(\mathbb{S}^{2}).$ The coherent states in %
\eqref{1.10} possess a form similar to \eqref{1.1} and will enables us,
starting from the observation made in \eqref{1.7}, to attach to each
eigenspace $\mathcal{A}_{B,m}\left( \mathbb{S}^{2}\right) $ a photon
counting probability distribution in the same way as for the space $\mathcal{%
A}_{B}(\mathbb{S}^{2})\equiv \mathcal{A}_{B,0}(\mathbb{S}^{2})$
through the quantities
\begin{align}
& p_{j}\left( 2B,\mu _{_{z}},m\right) =\frac{m!\left( 2B+m\right) !}{%
j!\left( 2B+2m-j\right) !}\mu _{_{z}}^{j-m}\left( 1-\mu _{_{z}}\right)
^{2B+m-j}\left( P_{m}^{\left( j-m,2B+m-j\right) }\left( 1-2\mu
_{_{z}}\right) \right) ^{2},  \label{1.11} \\
& j=0,1,\cdots ,2B+2m.  \notag
\end{align}
The latter can be considered as a kind of generalized binomial probability
distribution $X\sim \mathcal{B}\left( 2B,\mu _{_{z}},m\right) $ depending on
an additional parameter $m=0,1,2,\cdots $. Thus, we study the main
properties of the family\thinspace of probability distributions in %
\eqref{1.11} and we examine the quantum photon counting statistics with
respect\thinspace \ to the location in the Riemann sphere of the point $z$
labeling the generalized coherent states\ introduced formally in \eqref{1.10}.

The paper is organized as follows. In Section 2, we recall briefly the
principal statistical properties of the binomial states. Section 3 deals
with some needed facts on the Schrödinger operator with uniform magnetic
field on the Riemann sphere with an explicit description of the
corresponding eigenspaces. Section 4 is devoted to a coherent states
formalism. This formalism is applied so as to construct a set of generalized
coherent states attached to each spherical Landau level. In Section 5, we
introduce the announced generalized binomial probability distribution and we
give its main parameters. In section 6, we discuss the
classicality/nonclassicality of the generalized coherent states with respect
to the location of their labeling points belonging to the Riemann sphere.

\section{The binomial states}

The binomial states in their first form were introduced by Stoler \textit{et
al.} \cite{Stoler84} to define a pure state of a single mode of the
electromagnetic field for which the photon number\ density is binomial. Like
the generalized coherent states (whose the coefficients of its $j$ states
expansion are allowed to have additional arbitrary phases) a generalized
binomial state can be defined by
\begin{equation}
\mid n,\mu , \theta>=\sum\limits_{j=0}^{n}\left( \frac{n!}{j!\left( n-j\right) !}\mu
^{j}\left( 1-\mu \right) ^{n-j}\right) ^{\frac{1}{2}}e^{ij\theta }\mid j>
\label{2.1}
\end{equation}
and has as a photon counting probability
\begin{equation}
p_{j}\left( n,\mu \right) =\frac{n!}{j!\left( n-j\right) !}\mu ^{j}\left(
1-\mu \right) ^{n-j}  \label{2.2}
\end{equation}
which follows the binomial law $Y\sim $ $\mathcal{B}\left( n,\mu \right) $
with parameters $n$ and $\mu $; $n\in \mathbb{Z}_{+}$, $0<\mu <1.$ The
connection with our notations in \eqref{1.1} and \eqref{1.2} can be made by
setting $n=2B$, $z=\left| z\right| e^{i\theta }$ and $\left| z\right|
^{2}=\mu \left( 1-\mu \right) ^{-1}.$

Note that in limits $\mu \rightarrow 0$ and $\mu \rightarrow 1$ the binomial
state reduces to number states $\mid 0>$ and $\mid n>$ respectively. In a
different limit of $n\rightarrow +\infty \quad $ and $\quad \mu \rightarrow
0\quad $ with $\quad n\mu \rightarrow \lambda$, the probability distribution %
\eqref{2.2} goes to the Poisson distribution $\mathcal{P}(\lambda) $
\begin{equation}
p_{j}\left( \lambda \right) =\frac{\lambda ^{j}}{j!}e^{-\lambda },\qquad
j=0,1,2,\dots ,  \label{2.3}
\end{equation}
which characterize the coherent states of the harmonic oscillator. In fact,
and as pointed out in \cite{Stoler84} the binomial states interpolate
between \textit{number states} (nonclassical states) and \textit{coherent
states} (classical states). It partakes of the properties of both and
reduces to each in different limits.

The characteristic function of the random variable $Y\sim \mathcal{B}(n,\mu )
$ is given by
\begin{equation}
\mathcal{C}_{Y}(t)=\Big((1-\mu )+\mu e^{it}\Big)^{n}  \label{2.4}
\end{equation}
from which one obtains the mean value and the variance as
\begin{equation}
E(Y)=n\mu \quad \mbox{and}\quad Var(Y)=n\mu \left( 1-\mu \right)
\label{2.5}
\end{equation}
Therefore, the Mandel parameter, which measures deviation from the
Poissonian distribution,
\begin{equation}
Q=\frac{Var(Y)}{E(Y)}-1=-\mu,   \label{2.6}
\end{equation}
is always negative. Thus photon statistics in the the binomial states is
always \textit{sub-Poissonian}.

\begin{remark}
\ We should note that a binomial state also admits a \textit{ladder operator
definition } \cite{FuSasaki} which means that this state is an eigenstate of
a proper combination of the number operator and the annihilation operator
via the Holstein-Primakoff realization of the Lie algebra of the group $SU(2)
$.
\end{remark}

\section{An orthonormal basis of $\mathcal{A}_{B, m}(\mathbb{S}^{2}) $}

Let $\mathbb{S}^{2}\subset \mathbb{R}^{3}$\ denotes the unit sphere with the standard metric of constant Gaussian curvature $\kappa=1.$ We shall
identify the sphere $\mathbb{S}^{2}$ with the extended complex plane $%
\mathbb{C}\mathbf{\cup }\left\{ \infty \right\} $, called the Riemann
sphere, via the stereographic coordinate $z=x+iy$; $x, y\in \mathbb{R}.$ We
shall work within a fixed coordinate neighborhood with coordinate $z$
obtained by deleting the ''point at infinity'' $\{ \infty \}$. Near this
point we use instead of $z$\ the coordinate $z^{-1}.$

In the stereographic coordinate $z$, the Hamiltonian operator of the Dirac
monopole with charge $q=2B$ reads \cite[p.598]{{FeraVese}}:
\begin{equation}
H_{B}:=-\left( 1+\left| z\right| ^{2}\right) ^{2}\frac{\partial ^{2}}{%
\partial z\partial \overline{z}}-Bz\left( 1+\left| z\right| ^{2}\right)
\frac{\partial }{\partial z}+B\overline{z}\left( 1+\left| z\right|
^{2}\right) \frac{\partial }{\partial \overline{z}}+B^{2}\left( 1+\left|
z\right| ^{2}\right) -B^{2}.  \label{3.1}
\end{equation}
This operator acts on the sections of the $U(1)$-bundle with the first Chern
class $q$. We have denoted by $B\in \mathbb{Z}_{+}$\ the strength of the
quantized magnetic field. We shall consider the Hamiltonian $H_{B}$ in %
\eqref{3.1} acting in the Hilbert space $L^{2}(\mathbb{S}^{2}):=L^{2}\Big( %
\mathbb{S}^{2},(1+\left| z\right| ^{2})^{-2}d\nu (z)\Big)$, $d\nu (z)=\pi
^{-1}dxdy$ being the Lebesgue measure on $\mathbb{C}\mathbf{\equiv }%
\mathbb{R}^{2}$. Its spectrum\ consists on an infinite number of eigenvalues
(\textit{spherical Landau levels}) of the form
\begin{equation}
\epsilon _{m}^{B}:=\left( 2m+1\right) B+m\left( m+1\right) ,\quad
m=0,1,2,\cdots ,  \label{3.2}
\end{equation}
with finite degeneracy $2B+2m+1$ (see \cite[p.598]{FeraVese}). In order to
present expressions of the corresponding eigensections in the coordinate $z$,
we first mention that the shifted operator $H_{B}-B$ on $L^{2}(\mathbb{S}%
^{2})$ is intertwined with the invariant Laplacian
\begin{equation}
\Delta _{2B}:=-\left( 1+\left| z\right| ^{2}\right) ^{2}\frac{\partial ^{2}}{%
\partial z\partial \overline{z}}+2B\overline{z}\left( 1+\left| z\right|
^{2}\right) \frac{\partial }{\partial \overline{z}}  \label{3.3}
\end{equation}
acting in the Hilbert space $L^{2,B}(\mathbb{S}^{2}):=L^{2}\Big( \mathbb{S}%
^{2},(1+\left| z\right| ^{2})^{-2-2B}d\nu (z)\Big)$. Namely, we have
\begin{equation}
\mathsf{\ }H_{B}-B=\left( 1+\left| z\right| ^{2}\right) ^{-B}\Delta
_{2B}\left( 1+\left| z\right| ^{2}\right) ^{B},  \label{3.4}
\end{equation}
and therefore any ket $|\phi >$ of $L^{2,B}(\mathbf{S}^{2})$ is represented
by
\begin{equation}
\left( 1+\left| z\right| ^{2}\right) ^{-B}<z|\phi >\quad \text{in }\quad
L^{2}(\mathbf{S}^{2}).  \label{3.5}
\end{equation}
We denote by $\mathcal{A}_{B,m}(\mathbb{S}^{2})$\ the eigenspace of $H_{B}$
in $L^{2}(\mathbb{S}^{2})$, corresponding to the eigenvalue $\epsilon
_{m}^{B}$\ given in \eqref{3.2}. Then, by \cite{PZ} together with \eqref{3.5}
and the intertwining relation \eqref{3.4}, we obtain the following
orthogonal basis of $\mathcal{A}_{B,m}(\mathbb{S}^{2})$:
\begin{equation}
h_{j}^{B,m}(z):=\left( 1+\left| z\right| ^{2}\right) ^{-B}z^{j}\text{ }%
Q_{B,m,j}\left( \frac{\left| z\right| ^{2}}{1+\left| z\right| ^{2}}\right)
,\quad 0\leq j\leq 2B+2m,  \label{3.6}
\end{equation}
where $Q_{B,m,j}(\cdot )$\ is the polynomial function given by
\begin{equation}
Q_{B,m,j}\left( t\right) =t^{-j}\left( 1-t\right) ^{j-2B}\left( \frac{d}{dt}%
\right) ^{m}\left[ t^{j+m}\left( 1-t\right) ^{2B-j+m}\right] .  \label{3.7}
\end{equation}
According to the Jacobi's formula (\cite{Magnus}):
\begin{equation}
\Big(\frac{d}{dx}\Big)^{m}\Big(x^{c+m-1}(1-x)^{b-c}\Big)=\frac{\Gamma (c+m)}{%
\Gamma (c)}x^{c-1}(1-x)^{b-c-m}{_{2}F_{1}}(-m,b;c;x),  \label{3.8}
\end{equation}
${_{2}F_{1}}(a,b,c;x)$ being the Gauss hypergeometric function, it follows
that
\begin{equation}
\text{ }Q_{B,m,j}\left( t\right) =\frac{\left( m+j\right) !}{j!}{_{2}F_{1}}%
\left( -m,2B+m+1,j+1;t\right)   \label{3.9}
\end{equation}
The latter can also be expressed in terms of Jacobi polynomials via the
transformation (\cite[p.283]{Magnus})
\begin{equation}
_{2}F_{1}\left( k+\nu +\varrho +1,-k,1+\nu ;\frac{1-t}{2}\right) =\frac{%
k!\Gamma \left( 1+\nu \right) }{\Gamma \left( k+1+\nu \right) }P_{k}^{(\nu
,\varrho )}(t).  \label{3.10}
\end{equation}
So that the orthogonal basis in \eqref{3.6} reads simply in terms of Jacobi polynomial as
\begin{equation}
h_{j}^{B,m}(z)=m!\left( 1+\left| z\right| ^{2}\right)
^{-B}z^{j-m}P_{m}^{\left( j-m,2B+m-j\right) }\left( \frac{1-\left| z\right|
^{2}}{1+\left| z\right| ^{2}}\right) .  \label{3.11}
\end{equation}
Also, one obtains the norm square of the eigenfunction $h_{j}^{B,m}$ given
in \eqref{3.6} as
\begin{equation}
\rho _{B,m}\left( j\right) :=\text{ }\left| \left| h_{j}^{B,m}\right|
\right| _{L^{2}(\mathbb{S}^{2})}^{2}=\frac{m!(m+j)!(2B+m-j)!}{\left(
2B+2m+1\right) \left( 2B+m\right) !}.  \label{3.12}
\end{equation}
Finally, by Theorem 1 of \cite[p.231]{PZ} and thank to \eqref{3.5}, we
obtain the following expression for the reproducing kernel of the Hilbert
eigenspace $\mathcal{A}_{B,m}(\mathbb{S}^{2})$:
\begin{align}
K_{B,m}\left( z,w\right) =(2B+2m+1)& \frac{\left( 1+z\overline{w}\right)
^{2B}}{\left( 1+\left| z\right| ^{2}\right) ^{B}\left( 1+\left| w\right|
^{2}\right) ^{B}}  \label{3.13} \\
& \times _{2}F_{1}\left( -m,m+2B+1,1;\frac{\left| z-w\right| ^{2}}{\left(
1+\left| z\right| ^{2}\right) \left( 1+\left| w\right| ^{2}\right) }\right) .
\notag
\end{align}

\begin{remark}
Note that in the case $m=0,$ elements of the orthogonal basis reduce further
to $h_{j}^{B,0}(z)=(1+|z|^{2})^{-B}z^{j}$ and the reproducing kernel reads
simply as
\begin{equation*}
K_{B,0}(z,w)=(2B+1){(1+z\overline{w})^{2B}}{%
(1+|z|^{2})^{-B}(1+|w|^{2})^{-B}}.
\end{equation*}
\end{remark}

\begin{remark}
Note that in higher dimension, i.e., in the case of the $n$-dimensional projective space $\mathbb{P}(\C^n) (=S^1\backslash S^{2n+1})$ equipped
with the Fubini-Study metric, an explicit formulae for the reproducing kernels of the eigenspaces associated with the Schr\"odinger operator with constant magnetic field written in the local coordinates (of the chart $\C^n$) as 
\begin{equation} 
H_B := (1+|z|^2)  \left\{ \sum_{i,j=1}^n \left( \delta_{ij}+z_i \bar z_j\right)\frac{\partial^2}{\partial z_i\partial \bar z_j} -B\sum_{j=1}^n \left( z_j\frac{\partial}{\partial z_j} -\bar z_j\frac{\partial}{\partial \bar z_j}\right) -B^2\right\} + B^2
\end{equation}
have been obtained in  \cite{Hafoud}.
\end{remark}


\section{Generalized coherent states}

Let $(\mathcal{H}, \left<\cdot, \cdot\right>_{\mathcal{H}})$ be a finite $d$-dimensional functional Hilbert space with an orthonormal basis $\left\{
\phi _{n}\right\}_{n=1}^{d}$ and $\mathcal{A}^{2}$ a finite $d$-dimensional
subspace of the Hilbert space $L^{2}(\Omega, ds)$, of square integrable
functions on a given measured space $(\Omega, ds)$, with an orthogonal
basis $\left\{\Phi _{n}\right\}_{n=1}^{d}$. Then, $\mathcal{A}^{2}$ is a
reproducing kernel Hilbert space whose the reproducing kernel is given by
\begin{equation}  \label{4.1}
K(x, y):=\sum_{n=1}^{d}\frac{\Phi_{n}(x) \overline{\Phi_{n}(y)}}{\rho _{n}};
\qquad x, y\in \Omega,
\end{equation}
where we have set $\rho _{n}:=\left\|
\Phi_{n}\right\|_{L^{2}(\Omega,ds)}^{2}$. Associated to the data of $(%
\mathcal{A}^{2}, \Phi _{n})$ and $(\mathcal{H}, \phi _{n})$, we introduce
the following

\begin{definition}
We define the generalized coherent states to be the elements of $\mathcal{H}$
given by
\begin{equation}
\Phi _{x}:=\left( \omega _{d}(x)\right) ^{-\frac{1}{2}}\sum_{n=1}^{d}\frac{%
\Phi _{n}(x)}{\sqrt{\rho _{n}}}\phi _{n};\qquad x\in \Omega ,  \label{4.2}
\end{equation}
where $\omega _{d}(x)$ stands for $\omega _{d}(x):=K(x,x).$
\end{definition}

Note that the choice of the Hilbert space $\mathcal{H}$ defines a
quantization of $\Omega$ into $\mathcal{H}$ by considering the inclusion map $%
x\longmapsto \Phi _{x}$. Furthermore, it is straightforward to check that $%
\left\langle \Phi _{x},\Phi _{x}\right\rangle _{\mathcal{H}}=1$ and to show
that the corresponding coherent state transform (CST) $\mathcal{W}$ on $\mathcal{H}
,$
\begin{equation}
\mathcal{W}[f](x):=(\omega _{d}(x))^{\frac{1}{2}}\left\langle \Phi_{x},f\right\rangle _{\mathcal{H}}
;\qquad f\in \mathcal{H}, 
 \label{cst}
\end{equation}
 defines an isometry from $\mathcal{H}$ into $\mathcal{A}^{2}$. Thereby we
have a resolution of the identity, i.e., we have the following integral
representation
\begin{equation}
f(\cdot )=\int\limits_{\Omega }\left\langle \Phi _{x},f\right\rangle _{%
\mathcal{H}}\Phi _{x}(\cdot )\omega _{d}(x)ds (x)
;\qquad f\in \mathcal{H}.  \label{resId}
\end{equation}

\begin{remark}
Note that formula \eqref{4.2} can be considered as a generalization (in the
finite dimensional case) of the series expansion of the well-known canonical
coherent states
\begin{equation}
|\zeta >=\left( e^{\left| \zeta \right| ^{2}}\right) ^{-\frac{1}{2}%
}\sum_{k\geq 0}\frac{\zeta ^{k}}{\sqrt{k!}}\phi _{k}  \label{4.5}
\end{equation}
with $\phi _{k}:=|k>$\ being the number states of the harmonic oscillator.
\end{remark}

We can now construct for each spherical Landau level $\epsilon _{m}^{B}$
given in \eqref{3.2} a set of generalized coherent states (GCS) according to
formula \eqref{4.2} as
\begin{equation}
\vartheta _{z,B,m}\equiv \mid z,B,m>=\left( K_{B,m}\left( z,z\right) \right)
^{-\frac{1}{2}}\sum\limits_{j=0}^{2B+2m}\frac{h_{j}^{B,m}(z) }{\sqrt{\rho
_{B,m}\left( j\right) }}\mid \phi _{j}>  \label{4.6}
\end{equation}
with the following precisions:

\begin{itemize}
\item[$\bullet$]  $(\Omega ,ds ):=(\mathbb{S}^{2},\left( 1+\left| z\right|
^{2}\right) ^{-2}d\nu (z))$, $\mathbb{S}^{2}$ being identified with $%
\mathbb{C}\mathbf{\cup }\left\{ \infty \right\} $.

\item[$\bullet$]  $\mathcal{A}^{2}:=\mathcal{A}_{B,m}\left( \mathbb{S}%
^{2}\right) $ is the eigenspace of $H_{B}$ in $L^{2}(\mathbb{S}^{2})$ with
dimension $d_{B,m}=2B+2m+1$.

\item[$\bullet$]  $\omega (z)=K_{B,m}\left( z,z\right) =2B+2m+1$ (in view of %
\eqref{3.13}).

\item[$\bullet$]  $h_{j}^{B,m}(z)$ are the eigenfunctions given by %
\eqref{3.11} in terms of the Jacobi polynomials.

\item[$\bullet$]  $\rho _{B,m}\left( j\right) $\ being the norm square of $%
h_{j}^{B,m},$ given in \eqref{3.12}.

\item[$\bullet$]  $\mathcal{H}:=\mathcal{P}_{B+m}$ the space of
polynomials of degree less than $d_{B,m}$, which carries a unitary
irreducible representation of the compact Lie group $SU\left( 2\right) $
(see \cite{Vilenkin}). The scalar product in $\mathcal{P}_{B+m}$\ is written
as
\begin{equation}
\left\langle \psi ,\phi \right\rangle _{\mathcal{P}_{B+m}}=d_{B,m}\int%
\limits_{\mathbb{C}}ds (z)\left( 1+\left| z\right| ^{2}\right) ^{-2\left(
B+m\right) -2}\psi (z)\overline{\phi (z)}.  \label{4.7}
\end{equation}

\item[$\bullet$]  $\left\{ \phi _{j};0\leq j\leq 2B+2m\right\} $ is an
orthonormal basis of $\mathcal{P}_{B+m}$, whose elements are give explicitly
by:
\begin{equation}
\phi _{j}(\xi ):=\sqrt{\frac{(2\left( B+m\right) )!}{(2B+m-j)!(j+m)!}}\xi
^{j+m}.  \label{4.8}
\end{equation}
\end{itemize}

\begin{definition}
Wave functions of the GCS in \eqref{4.6} are expressed as
\begin{equation}
\vartheta _{z,B,m}\left( \xi \right) \equiv \left( 1+\left| z\right|
^{2}\right) ^{-B}\sum\limits_{j=0}^{2B+2m}\frac{\sqrt{m!\left( 2B+m\right)
!\left( 2B+2m\right) !}}{j!\left( 2B+2m-j\right) !}z^{j-m}P_{m}^{\left(
j-m,2B+m-j\right) }\left( \frac{1-\left| z\right| ^{2}}{1+\left| z\right|
^{2}}\right) \xi ^{j}  \label{4.9}
\end{equation}
\end{definition}

According to \eqref{resId}, the system of GCS $\mid \vartheta
_{z,B,m}>$ solves then the unity of the Hilbert space \ $\mathcal{P}_{B+m}$\
as
\begin{equation}
\mathbf{1}_{\mathcal{P}_{B+m}\mathsf{\ \ }}\mathbf{=}d_{B,m}\int\limits_{%
\mathbb{C}}d\nu (z)\left( 1+\left| z\right| ^{2}\right) ^{-2}\mid \vartheta
_{z,B,m}>\mathbf{<}\vartheta _{z,B,m}\mid \mathsf{.}  \label{4.10}
\end{equation}
They also admit a closed form \cite{Mouayn}, as
\begin{equation}
\vartheta _{z,B,m}\left( \xi \right) =\sqrt{\frac{\left( 2B+2m\right) !}{%
\left( 2B+m\right) !m!}}\left( \frac{\left( 1+\xi z\right) ^{2}}{1+\left|
z\right| ^{2}}\right) ^{B}\left( \frac{\left( \xi -\overline{z}\right)
\left( 1+\xi z\right) }{1+\left| z\right| ^{2}}\right) ^{m}.  \label{4.11}
\end{equation}

\begin{remark} Note that for $m=0,$  the expression above
reduces to
\begin{equation*}
<\xi |z,B,0>=\left( 1+\left| z\right| ^{2}\right) ^{-B}\left( 1+\xi z\right)
^{2B}.
\end{equation*}
which are wave functions of Perelomov's coherent states based on $SU(2)$ (see \cite[p.62]{perelomov}).
\end{remark}

\section{Generalized binomial probability distributions}

According to \eqref{4.2}, we see that the squared modulus of $\left\langle
\vartheta _{z,B,m},\phi _{j}\right\rangle _{\mathcal{H}}$, the projection of
coherent state $\vartheta _{z,B,m}$ onto the state $\phi _{j}$, is given by
\begin{equation}
\Big| \left\langle \vartheta _{z,B,m},\phi _{j}\right\rangle _{\mathcal{H}}%
\Big|^{2}=\bigg|\left( K_{B,m}\left( z,z\right) \right) ^{-\frac{1}{2}}\frac{%
h_{j}^{B,m}(z)}{\sqrt{\rho _{j}^{B,m}}}\bigg|^{2}=\frac{1}{\rho
_{j}^{B,m}d_{B,m}}\big|h_{j}^{B,m}(z)\big|^{2}.  \label{5.1}
\end{equation}
This is in fact the probability of finding $j$ photons in the coherent state
$\vartheta _{z,B,m}$. More explicitly, in view of \eqref{3.11}, the quantity
in \eqref{5.1} reads
\begin{equation}
\Big| \left\langle \vartheta _{z,B,m},\phi _{j}\right\rangle _{\mathcal{H}}%
\Big|^{2}=\frac{m!(2B+m)!}{j!(2B+2m-j)!}(1+|z|^{2})^{-2B}|z|^{2(j-m)}\bigg( %
P_{m}^{(j-m,2B+m-j)}\Big(\frac{1-|z|^{2}}{1+|z|^{2}}\Big)\bigg)^{2}.
\label{5.2}
\end{equation}
We denote the expression in \eqref{5.2} by $p_{j}(2B,\mu _{_{z}},m)$ for $%
j=0,1,2,\cdots $, with $\mu _{_{z}}={|}z{|^{2}}{(1+}|z{|^{2})^{-1}}$ or
equivalently $|z|^{2}=\mu _{_{z}}(1-\mu _{_{z}})^{-1}$. Motivated by quantum
probability, we then state the following

\begin{definition}
\label{BinProb} For fixed integers $B,m\in \Bbb{Z}_{+}$, the discrete
random variable $X$ having the probability distribution
\begin{equation}
p_{j}(2B,\mu _{_{z}},m)=\frac{m!(2B+m)!}{j!(2B+2m-j)!}\mu
_{_{z}}^{(j-m)}(1-\mu _{_{z}})^{2B+m-j}\Big( P_{m}^{(j-m,2B+m-j)}\big(1-2\mu
_{_{z}}\big)\Big)^{2},  \label{GBP}
\end{equation}
with $j=0,1,2,\cdots ,2B+2m$, and denoted by $X\sim \mathcal{B}(2B,\mu
_{_{z}},m)$, $0<\mu _{_{z}}<1$, will be called the generalized binomial
probability distribution associated to the weighted Hilbert space $\mathcal{A%
}_{B,m}(S^{2})$.
\end{definition}

\begin{remark}
Note that for $m=0$, the above expression in \eqref{GBP} reduces to
\begin{equation*}
p_{j}(2B,\mu _{_{z}},0)=\frac{(2B)!}{j!(2B-j)!}\mu _{_{z}}^{j}(1-\mu
_{_{z}})^{2B-j},\quad j=0,1,2,\cdots ,2B,
\end{equation*}
which is the standard binomial distribution with parameters $2B$ and $0<\mu
_{_{z}}<1.$,
\end{remark}

A convenient way to summarize all the properties of a probability
distribution $X$ is to explicit its characteristic function :
\begin{equation}
\mathcal{C}_{X}(t)=:E\left( e^{itX}\right) ,  \label{CharFct}
\end{equation}
where $t$ is a real number, $i:=\sqrt{-1}$ is the imaginary unit and $E$
denotes the expected value or the mean value. We precisely establish the
following result.

\begin{proposition}
For fixed $m=0,1,2,\cdots ,$ the characteristic function of $X\sim \mathcal{B%
}(2B,\mu _{_{z}},m)$ is given by
\begin{equation}
\mathcal{C}_{m}(t)=e^{imt}\Big([1-\mu _{_{z}}]+\mu _{_{z}}e^{it}\Big)^{2B}{P}%
_{m}^{(0,2B)}\big(1-4\mu _{_{z}}(1-\mu _{_{z}})(1-\cos (t))\big)
\label{CharFctExpl}
\end{equation}
for every $t\in \Bbb{R}$.
\end{proposition}

\begin{proof} Recall first that for every given fixed nonnegative integer $m$, the characteristic function $\mathcal{C}_{m}(t) $ in \eqref{CharFct} can be written as
\begin{align}
  \mathcal{C}_{m}(t) &=   \sum_{j=0}^{2B+2m}e^{ijt} p_{j}(2B,\mu _{_{z}},m) \nonumber\\ 
 & =\sum_{j=0}^{2B+2m}e^{ijt} \frac{m!(2B+m)!}{j!(2B+2m-j)!}
 \mu_{_z} ^{(j-m)} (1-\mu_{_z} )^{2B+m-j} \Big( P^{(j-m, 2B+m-j )}_m\big(1-2\mu_{_z} \big)\Big)^2 \label{ddd} ,
 \end{align}
according to the expression of $p_{j}(2B,\mu _{_{z}},m)$ given through \eqref{GBP}.
Next, by making the change $k=B+m-j$ in \eqref{ddd}, it follows 
\begin{equation}
 \mathcal{C}_{m}(t) = \sum_{k=-(B+m)}^{B+m}e^{i(B+m-k)t} \frac{m!(2B+m)!}{(B+m+k)!(B+m-k)!}
                              \mu_{_z} ^{B-k} (1-\mu_{_z} )^{B+k} \Big( P^{(B-k, B+k)}_m\big(1-2\mu_{_z} \big)\Big)^2. \label{Charact2}
\end{equation}
 Instead of the Jacobi polynomials,  it is convenient to consider the closely related
 function $\mathcal{P}_{r, s}^l(x)$ introduced in \cite[p. 270]{Vilenkin}. They can be defined through
 the formula \cite[Eq.1,  p. 288]{Vilenkin},
\begin{align}\label{Trans} {P}^{(r-s, r+s)}_{n-r}(x) = 2^r \bigg(\frac{(n-s)!(n+s)!}{(n-r)!(n+r)!}\bigg)^{1/2}(1-x)^{(s-r)/2}(1+x)^{-(s+r)/2}\mathcal{P}^n_{r, s}(x),
\end{align}
with $m=n-B$ (i.e.,  $n=B+m$) and $x=1-2\mu_{_z} $. We can then express the square of ${P}^{(B-k, B+k)}_{m}(x)$ as follows
$$ \Big({P}^{(B-k, B+k)}_{m}(1-2\mu_{_z} )\Big)^2
=  \frac{(B+m-k)!(B+m+k)!}{m!(2B+m)!}\mu_{_z} ^{-B+k}(1-\mu_{_z} )^{-B-k}
\left(\mathcal{P}^{B+m}_{B, k}(1-2\mu_{_z} )\right)^2.
$$
Therefore,  \eqref{Charact2} reduces further to
\begin{align}
 \mathcal{C}_{m}(t)  &= e^{i(B+m)t} \sum_{k=-(B+m)}^{B+m}e^{-ikt} \Big(\mathcal{P}^{B+m}_{B, k}(1-2\mu_{_z} )\Big)^2\\
 &\stackrel{(\star)}{=} (-1)^{B}e^{i(B+m)t} \sum_{k=-(B+m)}^{B+m}e^{-ik(t-\pi)} \mathcal{P}^{B+m}_{B, k}(1-2\mu_{_z} )\mathcal{P}^{B+m}_{k, B}(1-2\mu_{_z} ) \nonumber \\
 &= (-1)^{B}e^{i(B+m)t} e^{-iB(\varphi+\psi)} \mathcal{P}^{B+m}_{B, B}(\cos(\theta)). \label{ultra1}
 \end{align}
 The transition $(\star)$ above holds using the fact that (\cite[p.288]{Vilenkin})
 $$ \mathcal{P}^{l}_{j, k}(x)= (-1)^{j+k}\mathcal{P}^{l}_{k, j}(x). $$
 While the last equality can be checked easily using the addition formula  (\cite[Eq3,  p.326]{Vilenkin})
$$\sum_{k=-s}^{s}e^{-ik\tau } \mathcal{P}^{s}_{j, k}(\cos(\theta_1))\mathcal{P}^{s}_{k, l}(\cos(\theta_2)) = e^{-i(j\varphi+l\psi)} \mathcal{P}^{s}_{j, l}(\cos(\theta)).$$
Here the involved complex angles $\varphi$,  $\psi$ and $\theta$ are given through equations (8),  (8') and (8") in \cite[p. 270]{Vilenkin}. In our case,  they yield the followings
\begin{align}
  \cos(\theta)&= \cos^2(2\alpha) + \sin^2(2\alpha)\cos(t)  \\
  e^{i(\frac{\varphi+\psi}{2})}&=  \frac{-i( \cos^2(\alpha) + \sin^2(\alpha)e^{-it}) e^{it/2}}{\cos(\theta/2)}
\end{align}
for $\theta_1=\theta_2=2\alpha$,  so that
\begin{equation}
 e^{-iB(\varphi+\psi)} =   (-1)^B \Big(\cos(\theta/2)\Big)^{-2B} \Big(\cos^2(\alpha) + \sin^2(\alpha)e^{it} \Big)^{2B}e^{-iBt}.\label{Fact1}
\end{equation}
Next,  using the fact that $$2^{-s} (1+x)^{s}{P}^{(0, 2s)}_{n-s}(x) = \mathcal{P}^n_{s, s}(x), $$ which is a special case of
 \eqref{Trans},  with $s=B$,  $n-s=m$ and $x=\cos(\theta)$,  we obtain
\begin{equation}
\mathcal{P}^{B+m}_{B, B}(\cos(\theta)) = \Big(\cos(\theta/2)\Big)^{2B} {P}^{(0, 2B)}_{m}(\cos(\theta)).\label{Fact2}
\end{equation}
Finally,  by substituting \eqref{Fact1} and \eqref{Fact2} in \eqref{ultra1},  taking into account that $\sin^2(\alpha)=\mu_{_z} $ and $\cos^2(\alpha)=1-\mu_{_z} $,   we see that the characteristic function $\mathcal{C}_{m}(t) $ reads simply as
\begin{equation}
 \mathcal{C}_{m}(t)  = e^{imt} \Big(\cos^2(\alpha) + \sin^2(\alpha)e^{it} \Big)^{2B} {P}^{(0, 2B)}_{m}(\cos(\theta)) 
,   \label{ultra2}
\end{equation}
where $\cos(\theta)= 1-4\mu_{_z} (1-\mu_{_z} )(1-\cos(t))$.
\end{proof}

\begin{remark}
Note that by taking $m=0$ in \eqref{ultra2}, the characteristic function
reduces to
\begin{equation*}
\mathcal{C}_{Y}(t)=\Big(\cos ^{2}(\alpha )+\sin ^{2}(\alpha )e^{it}\Big)%
^{2B}=\Big([1-\mu _{_{z}}]+\mu _{_{z}}e^{it}\Big)^{2B}
\end{equation*}
which is the well-known characteristic function of the binomial random variable  $Y\sim \mathcal{%
B}(2B,\mu _{_{z}})$ with parameters $n=2B\in \Bbb{Z}_{+}$ and $0<\mu _{z}<1$ as in %
\eqref{2.4}.
\end{remark}

Now, as mentioned at the beginning of this section, the characteristic
function contains important information about the random variable $X$. For
example, various moments may be obtained by repeated differentiation of $%
\mathcal{C}_{m}(t)$ in \eqref{CharFctExpl} with respect to the variable $t$
and evaluation at the origin as
\begin{equation*}
E\left( X^{k}\right) =\left. \frac{1}{i^k}\frac{\partial^k }{\partial t^k}\left(C_{X}( t)\right) \right|_{t=0}.
\end{equation*}

\begin{corollary}
Let $m,2B\in \Bbb{Z}_{+}$. The mean value and the variance of $X\sim
\mathcal{B}(2B,\mu _{_{z}},m)$ are given respectively by
\begin{align}
E(X)& =m+2B\mu _{_{z}}  \label{mean} \\
Var(X)& =2B\mu _{_{z}}(1-\mu _{_{z}})+2\mu _{_{z}}(1-\mu _{_{z}})m(2B+m+1).
\label{var}
\end{align}
\end{corollary}

\begin{proof}
Let recall first that for every fixed integer $m=0,1,2, \cdots,$ we have
\begin{align*}
E(X) &=\left.\frac{\partial \mathcal{C}_{m}}{i\partial t}\right|_{t=0}\\
Var(X) &= E(X^2)- [E(X)]^2 = \left.\frac{\partial^2 \mathcal{C}_{m}}{i^2\partial t^2}\right|_{t=0} -
\left[ \left.\frac{\partial \mathcal{C}_{m}}{i\partial t}\right|_{t=0}\right]^2.
\end{align*}
Thus direct computation gives rise to
\begin{align*}\label{t}
\frac{\partial \mathcal{C}_{m}}{i\partial t}(t)= \Bigg[ m +  \frac{2B \mu_{_z}  e^{it}}{\Big([1-\mu_{_z} ] + \mu_{_z}  e^{it} \Big)}  - 4\mu_{_z} (1-\mu_{_z} )\sin(t) \Bigg(  \frac{\frac{\partial {P}^{(0, 2B)}_{m}(x)}{i\partial x}{|_{x=\cos(\theta)}}}{{P}^{(0, 2B)}_{m}(\cos(\theta))} \Bigg)  \Bigg] \mathcal{C}_{m}(t)
\end{align*}
and
\begin{align*}
\frac{\partial^2 \mathcal{C}_{m}}{i^2\partial t^2}(t)
&= \left[\frac{\partial}{i\partial t}\left( m +  \frac{2B\mu_{_z}  e^{it}}{\Big([1-\mu_{_z} ] + \mu_{_z}  e^{it} \Big)}  - 4\mu_{_z} (1-\mu_{_z} )\sin(t) \Bigg(  \frac{\frac{\partial {P}^{(0, 2B)}_{m}(x)}{i\partial x}{|_{x=\cos(\theta)}}}{{P}^{(0, 2B)}_{m}(\cos(\theta))} \Bigg)  \right) \right] \mathcal{C}_{m}(t)
\\ & +   \Bigg[ m +  \frac{2B \mu_{_z}  e^{it}}{\Big([1-\mu_{_z} ] + \mu_{_z}  e^{it} \Big)}  - 4\mu_{_z} (1-\mu_{_z} )\sin(t) \Bigg(  \frac{\frac{\partial {P}^{(0, 2B)}_{m}(x)}{i\partial x}{|_{x=\cos(\theta)}}}{{P}^{(0, 2B)}_{m}(\cos(\theta))} \Bigg)  \Bigg] \frac{\partial \mathcal{C}_{m}}{i\partial t}(t) \\
&=  \left[ \frac{2B\mu_{_z}  e^{it}\Big([1-\mu_{_z} ] + \mu_{_z}  e^{it} \Big) - 2B\mu_{_z} ^2 e^{it} }{\Big([1-\mu_{_z} ] + \mu_{_z}  e^{it} \Big)^2}
     + 4\mu_{_z} (1-\mu_{_z} )\cos(t) \Bigg(\frac{\frac{\partial {P}^{(0, 2B)}_{m}(x)}{\partial x}|_{x=\cos(\theta)}} {{P}^{(0, 2B)}_{m}(\cos(\theta))} \Bigg) \right] \mathcal{C}_{m}(t)  \\
&  - 4\mu_{_z} (1-\mu_{_z} )\sin(t) \frac{\partial}{i\partial t}\Bigg(  \frac{\frac{\partial {P}^{(0, 2B)}_{m}(x)}{i\partial x}|_{x=\cos(\theta)}}{{P}^{(0, 2B)}_{m}(\cos(\theta))} \Bigg) \mathcal{C}_{m}(t)  \\
\\ & +   \Bigg[ m +  \frac{2B \mu_{_z}  e^{it}}{\Big([1-\mu_{_z} ] + \mu_{_z}  e^{it} \Big)}  - 4\mu_{_z} (1-\mu_{_z} )\sin(t) \Bigg(  \frac{\frac{\partial {P}^{(0, 2B)}_{m}(x)}{i\partial x}|_{x=\cos(\theta)}}{{P}^{(0, 2B)}_{m}(\cos(\theta))} \Bigg)  \Bigg] \frac{\partial \mathcal{C}_{m}}{i\partial t}(t)
\end{align*}
To conclude,  we have to use successively the facts that for $t=0$,  we have $\cos(\theta)=1$ and  $\mathcal{C}_{m}(0)=1$, together with
$$
\frac{\partial {P}^{(a, b)}_{m}}{i\partial x}(x)= \frac {a+b+m+1}2{P}^{(a+1, b+1)}_{m-1}(x)
\quad \mbox{and} \quad {P}^{(a, b)}_{m}(1)=\frac{\Gamma(a+m+1)}{m!\Gamma(a+1)}.$$
Thus, we have
$$E(X)=\left.\frac{\partial \mathcal{C}_{m}}{i\partial t}\right|_{t=0}= \Big( m +  2B \mu_{_z}\Big) \mathcal{C}_{m}(0) =  m +  2B \mu_{_z}. $$
We have also
\begin{align*}
\frac{\partial^2 \mathcal{C}_{m}}{i^2\partial t^2}{|_{t=0}}
&=  \left[2B\mu_{_z} ( 1  - \mu_{_z} ) +  2\mu_{_z} (1-\mu_{_z} )(2B+m+1) \Bigg(\frac{{P}^{(1, 2B+1)}_{m-1}(1)} {{P}^{(0, 2B)}_{m}(1)} \Bigg)\right] \mathcal{C}_{m}(0)
\\ & +   \Big[ m +  2B \mu_{_z}   \Big] \frac{\partial \mathcal{C}_{m}}{i\partial t}{|_{t=0}}
\end{align*}
and therefore
\begin{align*}
Var(X) &= \left.\frac{\partial^2 \mathcal{C}_{m}}{i^2\partial t^2}\right|_{t=0}   -  \Big( \left.\frac{\partial \mathcal{C}_{m}}{i\partial t}\right|_{t=0}\Big)^2\\
& =2B\mu_{_z} ( 1  - \mu_{_z} ) +  2\mu_{_z} (1-\mu_{_z} )m(2B+m+1).
\end{align*}
\end{proof}

\begin{remark}
Note that by taking $m=0$ in \eqref{mean} and \eqref{var}, we recover the
standard values
\begin{equation}
E(Y)=2B\mu _{_{z}}\quad \mbox{and}\quad Var(Y)=2B\mu _{_{z}}(1-\mu
_{_{z}}) \label{EspVar0}
\end{equation}
of the binomial probability distribution as given in \eqref{2.5}.
\end{remark}

\section{Photon counting statistics}

For an arbitrary quantum state one may ask to what extent is
''non-classical'' in a sense that its properties differ from those of
coherent states? In other words, is there any parameter that may reflect the
degree on non-classicality of a given quantum state? In general, to define a
measure of non-classicality of a quantum states one can follow several
different approach. An earlier attempt to shed some light on the
non-classicality of a quantum state was pioneered by Mandel \cite{Mandel},
who investigated radiation fields and introduced the parameter
\begin{equation}
Q=\frac{Var(X) }{E(X) }-1,  \label{MandelPar}
\end{equation}
to measure deviation of the photon number statistics from the Poisson
distribution, characteristic of coherent states. Indeed, $Q=0$ characterizes
Poissonian statistics. If $Q<0$, we have \textit{sub-Poissonian} statistics
otherwise, statistics are \textit{super-Poissonian. }

In our context and for $m=0$, as mentioned in Section 2, the fact that the
binomial probability distribution has a negative Mandel parameter, according
to \eqref{EspVar0}, and thereby the binomial states obeys sub-Poissonian
statistics.

For $m\neq 0$, we make use of the obtained statistical parameters $E(X)$ and
$Var(X)$ to calculate Mandel parameter corresponding the random variable $%
X\sim \mathcal{B}(2B,\mu _{_{z}},m).$ The discussion with respect to the
sign of this parameter gives rise to the following statement:

\begin{proposition}
\label{statics} Let $m$ and $B$ be nonegative integers and set
\begin{equation}
r_{\pm }(B,m):=\bigg(1\pm \Big( 1-\frac{1}{m(2B+m)}\Big)^{1/2}\bigg)^{1/2}.
\label{m2a}
\end{equation}
Then, $r_{-}\left( B,m\right) \leq 1\leq r_{+}\left( B,m\right) $ and the
photon counting statistics are:

\begin{itemize}
\item[i)]  Sub-Poissonian for points $z$ such that $|z|<r_{-}\left(
B,m\right) $ and $|z|>r_{+}\left( B,m\right) $.

\item[ii)]  Poissonian for points $z$ such that $|z|=r_{-}\left( B,m\right) $
or $|z|=r_{+}\left( B,m\right) $.

\item[iii)]  Super-Poissonian for $z$ such that $r_{-}\left( B,m\right)
<|z|<r_{+}\left( B,m\right) $.
\end{itemize}
\end{proposition}

\begin{proof} Assume that $m\ne 0$. Making use of \eqref{mean} and
\eqref{var}, we see that the Mandel parameter corresponding to the random
variable $X\sim \mathcal{B}\left( 2B,\mu _{_{z}},m\right) $ can be written
as follows $Q(X)=-T_{m}(\mu _{_{z}})/(2B\mu _{_{z}}+m)$, where we have set
\begin{align}
T_{m}(\mu _{_{z}})& =2(B+m[2B+m+1])\mu _{_{z}}^{2}-2m[2B+m+1]\mu _{_{z}}+m
\label{Mandel} \\
& =\Big( \mu _{_{z}}-\frac{md_{B,m}}{2(B+md_{B,m})}\Big)^{2}-\frac{%
m(d_{B,m}-1)(2Bm+m^{2}-1)}{4(B+md_{B,m})^{2}}
\end{align}
with $d_{B,m}:=2B+2m+1$. Hence, it is clear that $T_{m}(\mu _{_{z}})=0$,
viewed as second degree polynomials in $\mu _{_{z}}$, admits exactly two
real solutions given by
\begin{equation}
  \mu _{_{z}}^{\pm }(B,m):=\frac{md_{B,m}}{2(B+md_{B,m})}\left( 1\pm \left( 1-%
\frac{2(B+md_{B,m})}{md_{B,m}^{2}}\right) ^{1/2}\right) .  \label{roots}
\end{equation}
Now, assertions $\left( i\right) ,\left( ii\right) $ and $\left( iii\right) $
follow by discussing the sign of the parameter $Q_{m}(\mu _{_{z}})$ (i.e.,
the sign of $-T_{m}(\mu _{_{z}})$) with respect to the modulus of $z\in %
\mathbb{C}$ $\cup $ $\{\infty \}$, keeping in mind that $|z|^{2}=\mu
_{_{z}}/(1-\mu _{_{z}})$.
\end{proof}

The figure below illustrates the quantum photon counting statistics with
respect to the location in the extended complex plane of the point $z$
labeling the generalized coherent states $\vartheta _{z,B,m}$ discussed in
Proposition \ref{statics}. Here $r_{\pm}:=r_{\pm }(B,m)$ are as in \eqref{m2a}.

\newpage
\vspace*{-2cm}
\begin{figure}[htp]
\centering
\includegraphics[width=0.7\textwidth, height=0.9\textwidth]{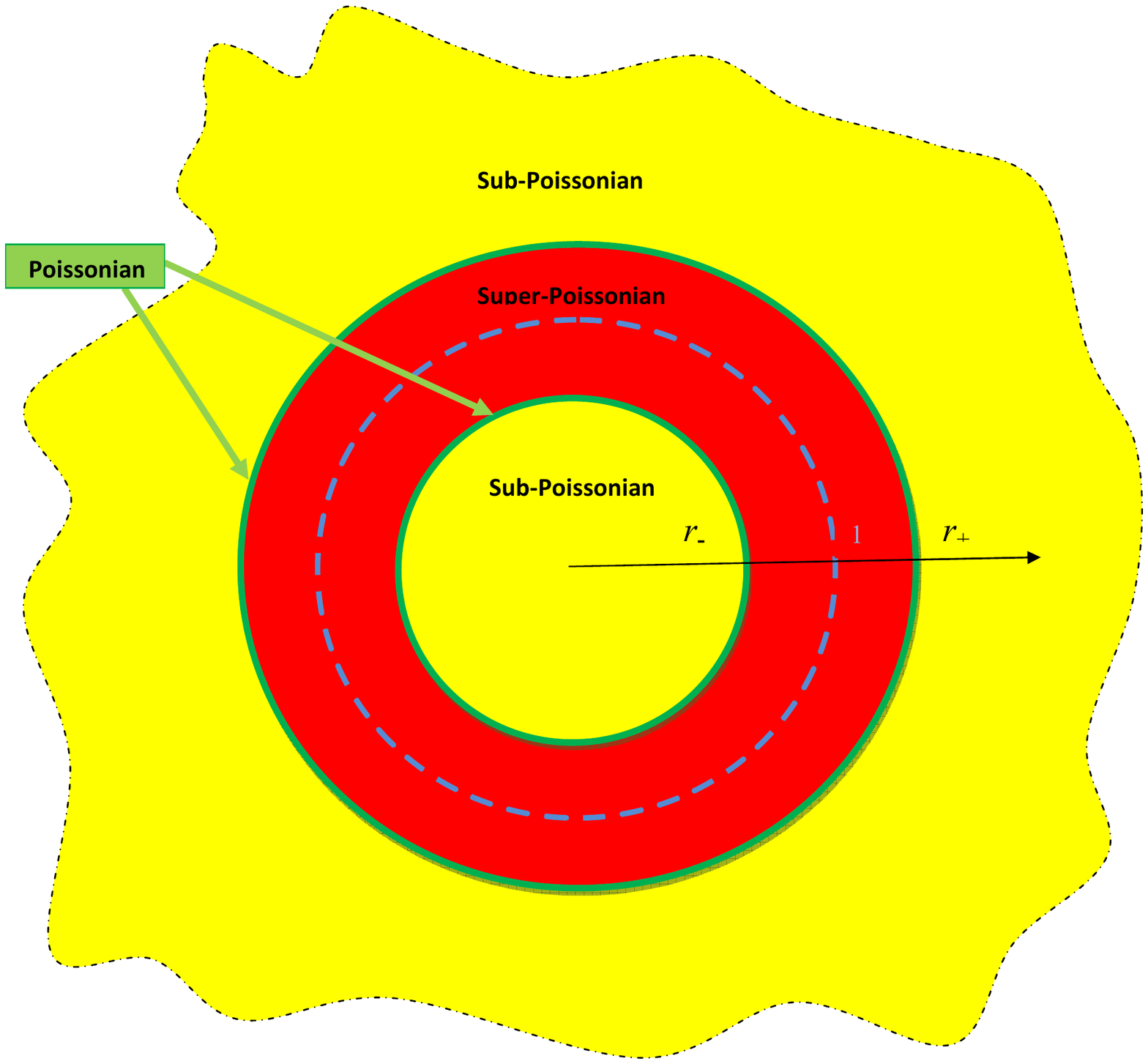}
\vspace*{-2cm}
\end{figure}


\end{document}